\documentclass[reqno,11pt]{amsart}

\usepackage{amssymb,amsmath,amsthm,amsfonts,mathtools,bbm,mathrsfs,enumitem,graphics,float,multirow}
\usepackage[normalem]{ulem}
\usepackage{hyperref}
\usepackage{comment}
\usepackage{hyperref}
\hypersetup{
    colorlinks,
    linkcolor={blue!80!black},
    citecolor={red!80!black},
    urlcolor={red!80!black}
}
\usepackage{xargs}                      
\usepackage{xcolor}  
\usepackage[colorinlistoftodos,prependcaption,textsize=tiny]{todonotes}
\newcommandx{\unsure}[2][1=]{\todo[linecolor=red,backgroundcolor=red!25,bordercolor=red,#1]{#2}}
\newcommandx{\change}[2][1=]{\todo[linecolor=blue,backgroundcolor=orange!25,bordercolor=blue,#1]{#2}}
\newcommandx{\info}[2][1=]{\todo[linecolor=OliveGreen,backgroundcolor=OliveGreen!25,bordercolor=OliveGreen,#1]{#2}}

\usepackage[top=1.6in,bottom=1.5in,left=1.3in,right=1.3in]{geometry}
\allowdisplaybreaks
\usepackage{cellspace}
\newcommand\redout{\bgroup\markoverwith{\textcolor{red}{\rule[0.5ex]{2pt}{0.8pt}}}\ULon}

\newtheorem{theorem}{Theorem}
\newtheorem{lemma}[theorem]{Lemma}
\newtheorem{claim}[theorem]{Claim}

\usepackage{tikz}
\usetikzlibrary{calc}

\tikzset{font= }

\newcommand\nc\newcommand
\nc\bfa{{\boldsymbol a}}\nc\bfA{{\boldsymbol A}}\nc\cA{{\mathscr A}}
\nc\bfb{{\boldsymbol b}}\nc\bfB{{\boldsymbol B}}\nc\cB{{\mathscr B}}
\nc\bfc{{\boldsymbol c}}\nc\bfC{{\boldsymbol C}}\nc\cC{{\mathscr C}}
\nc\bfd{{\boldsymbol d}}\nc\bfD{{\boldsymbol D}}\nc\cD{{\mathscr D}}
\nc\bfe{{\boldsymbol e}}\nc\bfE{{\boldsymbol E}}\nc\cE{{\mathscr E}}
\nc\bff{{\boldsymbol f}}\nc\bfF{{\boldsymbol F}}\nc\cF{{\mathscr F}}
\nc\bfg{{\boldsymbol g}}\nc\bfG{{\boldsymbol G}}\nc\cG{{\mathscr G}}
\nc\bfh{{\boldsymbol h}}\nc\bfH{{\boldsymbol H}}\nc\cH{{\mathscr H}}\nc\fH{{\mathfrak H}}
\nc\bfi{{\boldsymbol i}}\nc\bfI{{\boldsymbol I}}\nc\cI{{\mathcal I}}
\nc\bfj{{\boldsymbol j}}\nc\bfJ{{\boldsymbol J}}\nc\cJ{{\mathscr J}}
\nc\bfk{{\boldsymbol k}}\nc\bfK{{\boldsymbol K}}\nc\cK{{\mathscr K}}
\nc\bfl{{\boldsymbol l}}\nc\bfL{{\boldsymbol L}}\nc\cL{{\mathscr L}}
\nc\bfm{{\boldsymbol m}}\nc\bfM{{\boldsymbol M}}\nc\cM{{\mathscr M}}
\nc\bfn{{\boldsymbol n}}\nc\bfN{{\boldsymbol N}}\nc\sN{{\mathscr N}}
\nc\bfo{{\boldsymbol o}}\nc\bfO{{\boldsymbol O}}\nc\cO{{\mathscr O}}
\nc\bfp{{\boldsymbol p}}\nc\bfP{{\boldsymbol P}}\nc\cP{{\mathscr P}}
\nc\bfq{{\boldsymbol q}}\nc\bfQ{{\boldsymbol Q}}\nc\cQ{{\mathscr Q}}
\nc\bfr{{\boldsymbol r}}\nc\bfR{{\boldsymbol R}}\nc\cR{{\mathscr R}}
\nc\bfs{{\boldsymbol s}}\nc\bfS{{\boldsymbol S}}\nc\cS{{\mathscr S}}
\nc\bft{{\boldsymbol t}}\nc\bfT{{\boldsymbol T}}\nc\cT{{\mathscr T}}
\nc\bfu{{\boldsymbol u}}\nc\bfU{{\boldsymbol U}}\nc\cU{{\mathscr U}}
\nc\bfv{{\boldsymbol v}}\nc\bfV{{\boldsymbol V}}\nc\cV{{\mathscr V}}
\nc\bfw{{\boldsymbol w}}\nc\bfW{{\boldsymbol W}}\nc\cW{{\mathscr W}}
\nc\bfx{{\boldsymbol x}}\nc\bfX{{\boldsymbol X}}\nc\cX{{\mathscr X}}
\nc\bfy{{\boldsymbol y}}\nc\bfY{{\boldsymbol Y}}\nc\cY{{\mathscr Y}}
\nc\bfz{{\boldsymbol z}}\nc\bfZ{{\boldsymbol Z}}\nc\cZ{{\mathscr Z}}
\nc\pp{\mathbb{P}}
\nc\ee{\mathbb{E} }

\renewcommand{\le}{\leqslant}
\renewcommand{\leq}{\leqslant}
\renewcommand{\ge}{\geqslant}
\renewcommand{\geq}{\geqslant}

\nc{\Cay}{{\sf Cay}}

\nc{\ff}{{\mathbb F}}

\newcommand\remove[1]{}

\title[]{Sequence Reconstruction over the Deletion Channel}
\author[]{Fengxing Zhu}\thanks{Institute for Systems Research and Department of ECE, University of Maryland, College Park, MD 20742, USA, fengxing@terpmail.umd.edu. Supported in part by NSF grant CCF 2330909.}
\date{}

\begin{document}
\begin{abstract}
In this paper, we consider the Levenshtein's sequence reconstruction problem in the case where the transmitted codeword is chosen from $\{0,1\}^n$ and the channel can delete up to $t$ symbols from the transmitted codeword. We determine the minimum number of channel outputs (assuming that they are distinct) required to reconstruct a list of size $\ell-1$ of candidate sequences, one of which corresponds to the original transmitted sequence. More specifically,  we determine the maximum possible size of the intersection of 
$\ell \geq 3$ deletion balls of radius $t$ centered at $x_1, x_2, \dots, x_{\ell}$, where 
$x_i \in \{0,1\}^n$ for all $i \in \{1,2,\dots,\ell\}$ and $x_i \neq x_j$ for 
$i \neq j$, with $ n \geq t+\ell-1$ and $t \geq 1$.
\end{abstract}
\maketitle

\section{Introduction}
The study of sequence reconstruction was initiated by Levenshtein in 
\cite{Levenshtein,LEVENSHTEIN2001310}, where a sender transmits a codeword $x$ over multiple noisy channels. The receiver observes the outputs of these channels and attempts to uniquely reconstruct the transmitted codeword $x$. For the fixed codebook $\mathcal{C}$ and the channel the main task is to determine the minimum number of channel outputs required to guarantee unique reconstruction. The motivation for the sequence reconstruction problem originates from biology and chemistry, where traditional redundancy-based error correction methods are unsuitable. In recent years, the problem has regained attention due to its strong relevance to information retrieval in advanced storage technologies. In such systems, the stored data may consist of a single copy that is read multiple times or several redundant copies of the same information~\cite{Horovitz,Yaakoni2}. This problem~\cite{Horovitz} is particularly significant in the context of DNA data storage~\cite{Yazdi,Bornholt}, where numerous noisy copies of DNA strands are available, and the objective is to accurately reconstruct the original information from these imperfect observations.

 This sequence reconstruction problem has been extensively studied under various channels. Hirschberg and Regnier in \cite{Hirschberg} derived tight bounds on the number of string subsequences. In \cite{Levenshtein,LEVENSHTEIN2001310}, Levenshtein obtained the minimum number of channel outputs for the deletion and insertion channel required for unique reconstruction for the case where $\mathcal{C}$ consists of all binary vectors of length $n$. Gabrys and Yaakobi~\cite{Gabrys} later solved the sequence reconstruction problem over the $t$-deletion channel, where $\mathcal{C}$ consists of binary vectors such that $d_L(x,y) \geq 2$ for $x,y \in \mathcal{C}$, with $d_L(x,y)$ being the Levenshtein distance between $x$ and $y$. More recently, Pham, Goyal, and Kiah~\cite{PHAM2025105980} obtained a complete asymptotic solution for this problem where $\mathcal{C}$ consists of binary vectors such that $d_L(x,y) \geq \ell$ for $x,y \in \mathcal{C}$.

The sequence reconstruction problem has also been studied under other channel models 
besides the deletion channel. Sala, Gabrys, Schoeny, and Dolecek~\cite{Sala} solved 
this problem for the insertion channel, assuming the codebook $\mathcal{C}$ consists 
of $q$-ary vectors satisfying $d_L(x,y) \geq \ell$ for all distinct $x, y \in \mathcal{C}$. 
Abu-Sini and Yaakobi~\cite{Abu-Sini} investigated the reconstruction problem 
for channels involving a single deletion combined with multiple substitutions, as well 
as channels involving a single insertion combined with a single substitution.

A variant of the sequence reconstruction problem allows the decoder to output a list 
of possible sequences instead of a unique reconstruction. Yaakobi and 
Bruck~\cite{Yaakobi} studied this problem for channels introducing substitution errors. 
In particular, they investigated the maximum intersection of $m$ Hamming balls of 
radius $t$ centered at $x_1, x_2, \dots, x_m$, where $d_H(x_i, x_j) \geq d$ for 
$i \neq j$. Junnila, Laihonen, and Lehtil\"a~\cite{Junnila,Junnila2} analyzed the 
list size when the channel introduces substitution errors with $t = e + \ell$, where 
$e$ is the error-correcting capability of a binary code $\mathcal{C}$. More recently, 
they extended these results from the binary case to the $q$-ary case 
in~\cite{JUNNILA2025115279}.

In this paper we focus on the deletion channel. Formally, when a a codeword of length $n$ is sent through a $t$-deletion channel, a subsequence of length $n-t$ is received. A $t$-deletion correcting code $\mathcal{C}$ is a subset of length-$n$ binary vectors such that for any vector $x \in \mathcal{C}$, $x$ can be uniquely identified from any length-$(n-t)$ subsequence of $x$. More specifically, we study the minimum number of $t$-deletion channel outputs (assuming they are distinct) required to reconstruct a list of size $\ell-1$ of candidate sequences, one of which corresponds to the original transmitted sequence when $\mathcal{C}=\{0,1\}^n$. In other words, we determine the maximum possible size of the intersection of 
$\ell \geq 3$ deletion balls of radius $t$ centered at $x_1, x_2, \dots, x_{\ell}$, where 
$x_i \in \{0,1\}^n$ for all $i \in \{1,2,\dots,\ell\}$ and $x_i \neq x_j$ for 
$i \neq j$, with $n \geq t+\ell-1$ and $t \geq 1$.

Let 
$$N(n,\ell,t):=\max_{x_1 \neq x_2 ...\neq x_{\ell} ; x_1,x_2,...,x_{\ell} \in \mathbb{F}_2^n}|\mathcal{D}_t(x_1) \cap \mathcal{D}_t(x_2) ...\cap \mathcal{D}_t(x_{\ell})|,$$
where $\ell \geq 3$, $n \geq t+\ell-1$, and $t \geq 1$. We establish the following theorem. 
\begin{theorem} \label{theorem: main}
For $\ell \geq 3$, $n \geq t+\ell-1$, and $t\geq 1$ we have that 
\begin{align*}
    N(n,\ell,t)= \sum_{i=1}^{\ell-2}D(n-2i,t-i)+2D(n-2(\ell-1),t-(\ell-1)),
\end{align*}
where $D(n,t)=\sum_{i=0}^t \binom{n-t}{i}$. 
\end{theorem}

We adopt the techniques and analysis from \cite{Gabrys} and \cite{PHAM2025105980} to prove Theorem~\ref{theorem: main}. In particular, to prove the lower bound, we provide one construction of $\ell$ binary sequences attaining the bound when $1\leq t<\ell-1$, and a different construction when $t\geq \ell-1$. In order to prove the upper bound, we use induction similar to that in \cite{Gabrys} and \cite{PHAM2025105980}. Specifically, we will prove it by strong induction on $n+\ell+t$. The base cases are $t=0$ and $\ell=2$. 

\emph{Our Contributions.} Our first contribution is the study of the maximum size of the intersection of $\ell\geq 3$ deletion balls, for which we obtain an exact characterization. To the best of our knowledge, previous research papers in this direction have focused primarily on the intersection of two deletion balls, two insertion balls, two Hamming balls, or intersections involving different types of balls.

Our second contribution lies in the technical details of the proof. For the lower bound in the main theorem, the constraints on the parameters $n,\ell,t$ require two different constructions of $\ell$ binary sequences in order to attain the maximum intersection size. For the upper bound, the overall proof strategy is similar to those in \cite{Gabrys} and \cite{PHAM2025105980}. However, because of the parameter constraints in our setting, several boundary cases arising in the induction step must be treated separately, and their analysis is more involved than the corresponding arguments in \cite{Gabrys} and \cite{PHAM2025105980}.

\section{Preliminaries and Tools}
Let $x$ be a binary sequence of length $n$ over $\mathbb{F}_2^n$. The deletion ball of radius $t$ centered at $x \in \mathbb{F}_2^n$ is define to be 
$$ \mathcal{D}_t(x)=\{y \in \mathbb{F}_2^{n-t} | y \textrm{ is a subsequence of } x \}.$$

For any two sequences $x_1$ and $x_2$, their Levenshtein distance is $t$ if $\mathcal{D}_t(x_1) \cap \mathcal{D}_t(x_2) \neq \emptyset$ and $\mathcal{D}_{t-1}(x_1) \cap \mathcal{D}_{t-1}(x_2) = \emptyset$

Let $a_n \in \mathbb{F}_2^n$ be an alternating sequence where its first bit is 1. For $0 < t < n$, we denote the maximum size of a deletion ball of radius $t$, by $D(n,t)$,i.e 
$$D(n,t)=\max_{x \in \mathbb{F}_2^n}|\mathcal{D}_t(x)|.$$
From \cite{Levenshtein} and \cite{LEVENSHTEIN2001310}, we know that 
$$D(n,t)=| \mathcal{D}_t(a_n)|=\sum_{i=0}^t \binom{n-t}{i}$$
and also 
$$D(n,t)=D(n-1,t)+D(n-2,t-1).$$
Note that by convention $D(n,t)=0$ if $t <0$ or $n <t$. 

Due to Levenshtein in \cite{Levenshtein} and \cite{LEVENSHTEIN2001310}, for $1 < t<n$,
$$\max_{x_1 \neq x_2 ,x_1,x_2 \in \mathbb{F}_2^n}|\mathcal{D}_t(x_1) \cap \mathcal{D}_t(x_2)|=2D(n-2,t-1).$$

Before proceeding, we need to give some definitions and state some lemmas that will be used very often in our analysis.

Let $\chi \subset \mathbb{F}_2^n$ be a set and $v$ a sequence of length at most $n$. We denote by $\chi^{v}$ the set of all sequences in $\chi$ that start with the sequence $v$, that is, 
$$\chi^{v}=\{x \in \chi | v \textrm{ is  a prefix of }x\}.$$

For a sequence $v \in \mathbb{F}_2^m$ and a set $\chi \in \mathbb{F}_2^n$, the set $v \circ \chi$ is prepending the sequence $v$ before every sequence in $\chi$, that is 
$$v \circ \chi=\{(vx) | x \in \chi \}.$$
The following two lemmas will be used very often in our analysis. Lemma \ref{lemma 1} was derived in \cite{Gabrys} and Lemma \ref{lemma 2} was obtained in \cite{Liron}. 

\begin{lemma} \label{lemma 1}
Let $n$, $m_1$, and $t$ be positive integers, and $x=x^1x^2\cdots x^n \in \mathbb{F}_2^n$, $x_1 \in \mathbb{F}_2^{m_1}.$ Assume that $k$ is the smallest integer such that $x_1$ is a subsequence of $(x^1,x^2,...,x^{k})$. Then 
$$\mathcal{D}_t(x)^{x_1}=x_1 \circ \mathcal{D}_{t^*}(x^{k+1},...,x^n),$$
where $t^*=t-(k-m_1).$ In particular,
$$|\mathcal{D}_t(x)^{x_1}|=|\mathcal{D}_{t^*}(x^{k+1},...,x^n)|.$$
\end{lemma}

 \begin{lemma} \label{lemma 2}
 Let $l <n$, $x \in \mathbb{F}_2^n$, $y \in \mathbb{F}_2^{n-l}$, where $y \in \mathcal{D}_{l}(x)$ and $l <t$. Then $\mathcal{D}_{t-l}(y) \subset \mathcal{D}_{t}(x)$.
\end{lemma}

For the ease of notation, let us define
$$N_{\ell}(n,t):=\sum_{i=1}^{\ell-2}D(n-2i,t-i)+2D(n-2(\ell-1),t-(\ell-1)).$$

\section{The Intersection of $\ell \geq 3$ Deletion Balls}
In this section we will prove Theorem ~\ref{theorem: main}, that is for $\ell \geq 3$, $n \geq t+\ell-1$, and $t\geq 1$
$$N(n,\ell,t)=N_{\ell}(n,t):=\sum_{i=1}^{\ell-2}D(n-2i,t-i)+2D(n-2(\ell-1),t-(\ell-1)).$$

\subsection{The Lower Bound}

We will show that $N_{\ell}(n,t)$ is a lower bound for $N(n,\ell,t)$. In particular, we give two different constructions achieving the maximum intersection size according to whether $ 1 \leq t < \ell-1$ or $t \geq \ell-1$. 
\begin{theorem}
For $\ell \geq 3$, $n \geq t+\ell-1$, and $t\geq 1$, we have
$$N(n,\ell,t) \geq \sum_{i=1}^{\ell-2}D(n-2i,t-i)+2D(n-2(\ell-1),t-(\ell-1))$$
\end{theorem}

\begin{proof}
We will split the proof into two cases, namely, $ 1 \leq t < \ell-1$ and $t \geq \ell-1$. 

First let us assume that $t \geq \ell-1$. 
Let $x_1=10a_{n-2}$, $x_2=01a_{n-2}$, $x_j=\underbrace{0101\cdots01}_{2(j-1)}a_{n-2(j-1)}$  for $j \in \{2,\cdots,\ell\}$ and let $\chi=\cap_{i=1}^{\ell}\mathcal{D}_t(x_i)$. Observe that 
$$(0 \circ \mathcal{D}_{(t-1)}(x_1^3,...,x_1^n)) \cap (0 \circ \mathcal{D}_{t}(x_2^2,...,x_2^n)) \cap \cdots \cap (0 \circ \mathcal{D}_{t}(x_{\ell}^2,...,x_{\ell}^n))=0 \circ \mathcal{D}_{(t-1)}(x_1^3,...,x_1^n)$$
and 
\begin{align*}
 (1 \circ \mathcal{D}_{t}(x_1^2,...,x_1^n)) \cap (1 \circ \mathcal{D}_{t-1}(x_2^3,...,x_2^n)) \cap \cdots \cap (1 \circ \mathcal{D}_{t-1}(x_{\ell}^3,...,x_{\ell}^n))\\
=(1 \circ \mathcal{D}_{t-1}(x_2^3,...,x_2^n)) \cap (1 \circ \mathcal{D}_{t-1}(x_3^3,...,x_3^n)) \cdots \cap (1 \circ \mathcal{D}_{t-1}(x_{\ell}^3,...,x_{\ell}^n))
\end{align*}

By Lemma~\ref{lemma 1} and Lemma~\ref{lemma 2}, we have 
\begin{align*}
|\chi^{0}|& = |(0 \circ \mathcal{D}_{(t-1)}(x_1^3,...,x_1^n)) \cap (0 \circ \mathcal{D}_{t}(x_2^2,...,x_2^n)) \cap \cdots \cap (0 \circ \mathcal{D}_{t}(x_{\ell}^2,...,x_{\ell}^n))|\\
&=|(0 \circ \mathcal{D}_{t-1}(x_1^3,...,x_1^n))|\\
& = D(n-2,t-1),
\end{align*}
and 
\begin{align*}
|\chi^{1}|& = |(1 \circ \mathcal{D}_{t}(x_1^2,...,x_1^n)) \cap (1 \circ \mathcal{D}_{t-1}(x_2^3,...,x_2^n)) \cap (1 \circ \mathcal{D}_{t-1}(x_3^3,...,x_3^n)) \cap \cdots \cap (1 \circ \mathcal{D}_{t-1}(x_{\ell}^3,...,x_{\ell}^n))|\\
&=|(1 \circ \mathcal{D}_{t-1}(x_2^3,...,x_2^n)) \cap (1 \circ \mathcal{D}_{t-1}(x_3^3,...,x_3^n)) \cdots \cap (1 \circ \mathcal{D}_{t-1}(x_{\ell}^3,...,x_{\ell}^n))|\\
& = N_{\ell-1}(n-2,t-1).
\end{align*}
Therefore,
\begin{align*}
|\chi|& =|\chi^{0}|+|\chi^{1}|\\
& = D(n-2,t-1)+N_{\ell-1}(n-2,t-1).
\end{align*}

Now we have the following recursive relations. 
$$N_2(n,t)=2D(n-2,t-1)$$
$$N_3(n,t)=D(n-2,t-1)+2D(n-4,t-2)$$
$$N_4(n,t)=D(n-2,t-1)+N_3(n-2,t-1)$$
...
$$N_{\ell}(n,t)=D(n-2,t-1)+N_{\ell-1}(n-2,t-1).$$

Therefore, 
$$N_{\ell}(n,t)=\sum_{i=1}^{\ell-2}D(n-2i,t-i)+2D(n-2(\ell-1),t-(\ell-1)).$$

Now we consider the case where $ 1 \leq t < \ell-1$. Note that 
\begin{align}
N_\ell(n,t)=\sum_{i=1}^{t}D(n-2i,t-i) \label{eq: D(n-2i,t-i)}
\end{align}

Now we claim that 
\begin{align}
\sum_{i=1}^{t}D(n-2i,t-i)=D(n-1,t-1). \label{eq:D(n-1,t-1)}
\end{align}
Indeed, this follows by induction on $t$. Define $S(n,t):=\sum_{i=1}^{t}D(n-2i,t-i)$. Then We have the following induction statement. 
$$P(t): \; \; \; S(n,t)=D(n-1,t-1) \; \; \text{for all} \; n \geq 2t.$$

For $t=1$ and for all $n \geq 2$, $D(n-2,0)=D(n-1,0)=1$.  

Now suppose $P(t-1)$ holds, i.e., for all $m \geq 2(t-1)$, we have 
$$S(m,t-1)=D(m-1,t-2).$$

Now fix any $n \geq 2t$. Then
\begin{align*}
S(n,t)&=D(n-2,t-1)+\sum_{i=2}^tD(n-2i,t-i)\\
&= D(n-2,t-1)+ \sum_{j=1}^{t-1}D((n-2)-2j,(t-1)-j)\\
&=D(n-2,t-1)+S(n-2,t-1).
\end{align*}
Since $n-2\geq 2(t-1)$, the induction hypothesis $P(t-1)$ applies with $m=n-2$. By induction $S(n-2,t-1)=D(n-3,t-2)$ and then we have 
$$S(n,t)=D(n-2,t-1)+D(n-3,t-2)=D(n-1,t-1).$$ 

Let $b=a_{n-1}$ and recall $a_{n-1}$ is an alternating binary sequence of length $n-1$ where its first bit is $1$. We will construct $\ell$ binary sequences by using $b$. More specifically, if $b=b_1b_2\cdots b_{n-1}$, then define for each $i \in [n-1]$
 $$x_i:=b_1b_2\cdots b_{i-1}b_ib_ib_{i+1} \cdots b_{n-1}.$$
Since $n \geq t+\ell-1 \geq \ell$, so far we have constructed at least $\ell-1$ distinct binary sequences of length $n$. In addition, take $a_n$ and thus we have al least $\ell$ distinct binary sequences of length $n$. Since $n \geq \ell$, choose $\ell$ of them and call them $x_1,x_2,\cdots, x_{\ell}$. 

Now take $y \in \mathcal{D}_{t-1}(a_{n-1})$. For every $i$, we can obtain $y$ from $x_i$ by first deleting the one symbol that gives $a_{n-1}$, and then deleting the $t-1$ symbols that give $y$. Therefore, $y \in D_t(x_i)$ for every $i \in [\ell]$. Thus we have 
$$\mathcal{D}_{t-1}(a_{n-1}) \subset \cap_{i=1}^{\ell}\mathcal{D}_t(x_i).$$
Since $a_{n-1}$ is an alternating binary sequence of length $n-1$, we have $|\mathcal{D}_{t-1}(a_{n-1})|=D(n-1,t-1)$. 

Finally, we have 
\begin{align}
    N(n,\ell,t)& \geq |\cap_{i=1}^{\ell} \mathcal{D}_t(x_i)|\\
    & \geq D(n-1,t-1)\\
    & = N_{\ell}(n,t) \label{eq:N_l(n,t)},
\end{align}
where (\ref{eq:N_l(n,t)}) is valid by (\ref{eq: D(n-2i,t-i)}) and (\ref{eq:D(n-1,t-1)}). 
\end{proof}

\subsection{The Upper Bound}

 We will show that $N_{\ell}(n,t)$ is an upper bound for $N(n,\ell,t)$ by induction on $\ell+n+t$. 

 Let $$P(n,t,\ell): \; \; \max |\cap_{i=1}^\ell \mathcal{D}_t(x_i)| \leq N_{\ell}(n,t)$$ for all distinct $x_1,x_2,\cdots,x_{\ell} \in \{0,1\}^n$ under $n \geq t+\ell-1$.

 Let us first address the base cases. Now let $t=0$ and $\cap_{i=1}^{\ell}D_t(x_i)=\emptyset$ since all $x_i$s are distinct.

Now let $\ell=2$. We have $\max_{x_1,x_2}|\mathcal{D}_t(x_1)\cap \mathcal{D}_t(x_2)|=2D(n-2,t-1)$. Note that $N_2(n,t)=2D(n-2,t-1)$.

Now we will move onto the induction step. We can assume that for $t\geq 1$, $\ell \geq 3$, and $n \geq t+\ell-1$, $P(n_0,\ell_0,t_0)$ holds for every triple $(n_0,\ell_0,t_0)$ satisfying $t_0 \geq 0$, $\ell_0 \geq 2$, $n_0 \geq t_0+\ell_0 -1$, and $n_0+\ell_0+t_0 < n+t+\ell$. 

We will need a few lemmas to complete this step.

The proof strategy is as follows. We first show that if the first symbols of the $\ell$ sequences are all the same, then the size of the intersection of their deletion balls is at most $N_{\ell}(n,t)$. It therefore remains to consider the case in which the first symbols of the $\ell$ sequences are not all the same. We divide this case into the following three subcases:

\textbf{Subcase 1:}
At least one of the $\ell$ sequences begins with $00$, and at least one begins with $11$.

\textbf{Subcase 2:}
None of the $\ell$ sequences begins with $00$, at least one begins with $11$, and at least one begins with $01$.

\textbf{Subcase 3:}
Among the $\ell$ sequences, $\ell_1$ begin with $01$ and $\ell_2$ begin with $10$, where $\ell=\ell_1+\ell_2$.

Let us first show the following two lemmas which will be used in our analysis several times.
\begin{lemma} \label{lemma: D(x-1,y-1)}
For any integers $x,y$ and any $m\ge 2$,
$$ D(x-1,y-1)\ \le\ N_{m}(x,y).$$

\end{lemma}

\begin{proof}
We argue by induction on $m\ge 2$.

\medskip
For the base case let $m=2$. 
Here $N_2(x,y)=2D(x-2,y-1)$. 

We have 
$$
D(x-1,y-1)=D(x-2,y-1)+D(x-3,y-2).
$$
Moreover,
$$
D(x-3,y-2)\ \le\ D(x-2,y-1),
$$
Thus
$$
D(x-1,y-1)\ \le\ 2D(x-2,y-1)\ =\ N_2(x,y).
$$

For the induction step assume the claim holds for some $m-1\ge 2$, i.e.,
$$
D(u-1,v-1)\ \le\ N_{m-1}(u,v)\qquad\text{for all }u,v\in\mathbb{Z}.
$$
We show it for $m$. Using the same recursion,
$$
D(x-1,y-1)=D(x-2,y-1)+D(x-3,y-2).
$$
On the other hand, by the definition of $N_m$,
$$
N_m(x,y)=D(x-2,y-1)+N_{m-1}(x-2,y-1).
$$

Applying the induction hypothesis gives
$$
D(x-3,y-2)\ \le\ N_{m-1}(x-2,y-1).
$$
Therefore
$$
D(x-1,y-1)
\ \le\ D(x-2,y-1)+N_{m-1}(x-2,y-1)
\ =\ N_m(x,y),
$$
completing the induction.
\end{proof}

\begin{lemma} \label{lemma: N_m,N_L}
For any $m$ and any $L\ge m$,
$$
N_m(u-1,v-1)\ \le\ N_L(u,v).
$$
\end{lemma}
\begin{proof}
Let $L=m$. Since $D(n,t) \geq D(n-1,t-1)$, we have 
$$
N_m(u-1,v-1) \leq N_L(u,v). 
$$
Now assume that $L>m$.  We will bound $N_m(u-1,v-1)$ from above termwise.

For $1\le i\le m-2$, we have:
$$
D(u-1-2i,\,v-1-i)\ \le\ D(u-2i,\,v-i).
$$
Since $L> m$, the indices $i=1,\dots,m-2$ all lie in the index
range $1\le i\le L-2$ of $N_L(u,v)$. Thus each of these terms is dominated by
a corresponding term of $N_L(u,v)$.

Let $a=u-1-2(m-1)$ and $b=v-1-(m-1)$. Then
$$
2D(a,b)=2D\bigl(u-1-2(m-1),\,v-1-(m-1)\bigr).
$$

If $L>m$, then $m-1\le L-2$. We have 
$$
2D(a,b)
\le D\bigl(u-2(m-1)+1,\,v-(m-1)\bigr)
\le D\bigl(u-2(m-1),\,v-(m-1)\bigr).
$$

\medskip
\noindent
Therefore, we have 
$$
N_m(u-1,v-1)\ \le\ N_L(u,v). \qedhere
$$
\end{proof}

The next lemma shows that if the first symbols of the $\ell$ sequences are all the same, then the size of the intersection of their deletion balls is at most $N_{\ell}(n,t)$. 

Since the proof of this lemma is somewhat involved, we first describe the main ideas. When $n>t+\ell-1$, we partition the sequences in the intersection of the deletion balls according to their first symbols, after which the induction hypothesis can be applied directly.

However, in the boundary case $n=t+\ell-1$, a direct application of the induction hypothesis is no longer sufficient or is invalid. We therefore divide the proof into three cases: $t\leq \ell-3$, $t=\ell-2$, and $t\geq \ell-1$. When $t\geq \ell-1$, the size of the intersection of the deletion balls can be bounded trivially. When $t\leq \ell-3$, the desired bound follows by applying the induction hypothesis to only $\ell-1$ of the sequences. The remaining case $t=\ell-2$ requires a more delicate argument. In this case, bounding separately the portions of the intersection corresponding to the two possible first symbols is not sufficient. Instead, we couple the two parts and analyze them jointly. More precisely, we show that if the number of sequences in the intersection beginning with $a$ attains its maximum possible value, then the number of sequences in the intersection beginning with the other symbol must be smaller than a certain quantity, which completes the proof.

\begin{lemma} \label{lemma:x_1=x_2=a}
 Let $x_1$,$x_2$,$\cdots$ and $x_{\ell}$ be $\ell$ arbitrary sequences in $\mathbb{F}_2^n$ such that $x_i \neq x_j$ for $i \neq j$ and $a=x_1^1=x_2^1=\cdots =x_{\ell}^1$. Then 
$$|\cap_{i}^{\ell}\mathcal{D}_t(x_i) |\leq N_{\ell}(n,t).$$
\end{lemma}
\begin{proof}
    Let $\chi=\mathcal{D}_t(x_1) \cap \mathcal{D}_t(x_2) \cap \cdots \cap \mathcal{D}_t(x_{\ell})$. We have 
\begin{align}
|\chi^{a}|& =|a \circ \mathcal{D}_{t}(x_1^2,...,x_1^n)\cap \mathcal{D}_{t}(x_2^2,...,x_2^n) \cap \cdots \cap \mathcal{D}_{t}(x_{\ell}^2,...,x_{\ell}^n)|\\
&\leq \max_{x_1 \neq x_2 \neq \cdots x_{\ell}; x_1,x_2,\cdots, x_{\ell} \in \mathbb{F}_2^{n-1}}|\mathcal{D}_t(x_1) \cap \mathcal{D}_{t}(x_2)\cap \cdots \cap \mathcal{D}_{t}(x_{\ell})| \label{eq: intersection of D_t(x_i)}\\
& \leq N_{\ell}(n-1,t), \label{eq: N_l(n-1,t)}
\end{align}
where ($\ref{eq: N_l(n-1,t)}$) is by the induction hypothesis if $n> t+\ell-1$.

Suppose $x_1^k$ is the first occurrence of the symbol $\bar{a}$ in $x_1$ and the symbol $\bar{a}$ appears in $x_1$ not after it appears in $x_2$, $x_3$, $\cdots$ and $x_\ell$.

If $x_1^k=x_2^k=\cdots =x_{\ell}^k=\bar{a}$ we have 
\begin{align}
|\chi^{\bar{a}}|& =|\bar{a} \circ \mathcal{D}_{t-(k-1)}(x_1^{k+1},...x_1^n)\cap \mathcal{D}_{t-(k-1)}(x_2^{k+1},...,x_2^n) \cap \cdots \cap \mathcal{D}_{t-(k-1)}(x_{\ell}^{k+1},...,x_{\ell}^n)|\\
&\leq \max_{\mathbf{x_1} \neq \mathbf{x_2} \cdots \neq \mathbf{x_{\ell}};x_1, x_2,\cdots,x_{\ell}\in \mathbb{F}_2^{n-k}}|\mathcal{D}_{t-(k-1)}(x_1) \cap \mathcal{D}_{t-(k-1)}(x_2)\cap \cdots \cap \mathcal{D}_{t-(k-1)}(x_{\ell})|\\
& \leq N_{\ell}(n-k,t-k+1) \label{eq:N_l(n-k,t-k+1)} \\ 
& \leq N_{\ell}(n-2,t-1), \label{eq: N_l(n-2,-1)}
\end{align}
where (\ref{eq:N_l(n-k,t-k+1)}) is valid by the induction hypothesis if $n> t+\ell-1$ and (\ref{eq: N_l(n-2,-1)}) is valid due to Lemma ~\ref{lemma: N_m,N_L}.

If one of $x_2^k$, $x_3^k$ $\cdots$ or $x_{\ell}^k$ is equal to $a$, say $x_2^k=a$ then we have 
\begin{align}
|\chi^{\bar{a}}|& \leq |\bar{a} \circ  \mathcal{D}_{t-k}(x_2^{k+2},...x_2^n)|\\
&\leq D(n-k-1,t-k)\\
& \leq D(n-3,t-2)\\
& \leq N_{\ell}(n-2,t-1) \label{eq: D(n-3,t-2) < N},
\end{align}
where (\ref{eq: D(n-3,t-2) < N}) is valid due to Lemma~\ref{lemma: D(x-1,y-1)}.

Since $|\chi^a|+|\chi^{\bar{a}}|\leq N_{\ell}(n-1,t)+N_{\ell}(n-2,t-1)=N_{\ell}(n,t)$, the proof is done for $n > t+\ell-1$.

Now we have to address the boundary case, where $n=t+\ell-1$. From now on in this proof we assume that $n=t+\ell-1$. Note that, in this case, it is not valid to bound $|\chi^a|$ in the same way as in (\ref{eq: N_l(n-1,t)}) and $|\chi^{\bar a}|$ in the same way as in (\ref{eq: N_l(n-2,-1)}).

First suppose that $t \geq \ell-1$. Then trivially, we have 
$$|\cap_{i=1}^{\ell} \mathcal{D}_t(x_i)|\leq 2^{\ell-1}$$ and 
$N_\ell(n,t)=2^{\ell-1}$. 

Now suppose that $t \leq \ell-3$. It is easy to see we have $|\chi^a| \leq N_{\ell-1}(n-1,t)$ by only considering $\ell-1$ sequences on (\ref{eq: intersection of D_t(x_i)}) and by induction on $n+t+\ell$. Similarly we have 
$|\chi^{\bar{a}}| \leq N_{\ell-1}(n-2,t-1)$. Note that we have 
$$N_\ell(n,t)=D(n-2,t-1)+N_{\ell-1}(n-2,t-1).$$
Thus it suffices to show that $N_{\ell-1}(n-1,t) \leq D(n-2,t-1)$. For $t \leq \ell-3$ and by the recurrence $D(a,b)=D(a-1,b)+D(a-2,b-1)$, we have
$$N_{\ell-1}(n-1,t)=\sum_{i=1}^tD(n-1-2i,t-i)=D(n-2,t-1).$$

Now assume that $t=\ell-2$. Let $m:=\ell-2$. Then $t=m$, $\ell=m+2$, and $n=2m+1$. 

Then we have 
\begin{align*}
|\chi^{a}|& =|a \circ \mathcal{D}_{t}(x_1^2,...,x_1^n)\cap \mathcal{D}_{t}(x_2^2,...,x_2^n) \cap \cdots \cap \mathcal{D}_{t}(x_{\ell}^2,...,x_{\ell}^n)|\\
&\leq \max_{x_1 \neq x_2 \neq \cdots x_{\ell}; x_1,x_2,\cdots, x_{\ell} \in \mathbb{F}_2^{n-1}}|\mathcal{D}_t(x_1) \cap \mathcal{D}_{t}(x_2)\cap \cdots \cap \mathcal{D}_{t}(x_{\ell})| \\
& \leq 2^m
\end{align*}
and by the induction hypothesis we have 
$$|\chi^{\bar{a}}| \leq N_{\ell-1}(n-2,t-1).$$

Note that $N_{\ell-1}(n-2,t-1)=D(2m-2,m-2)=2^m-m-1$ and $N_{\ell}(n,t)=D(2m,m-1)=2^{m+1}-m-2$. Now we show the following statement. Let $x_1,x_2,\cdots, x_{m+2}$ be binary sequences of length $2m+1$ with $x_1^1=x_2^1=\cdots=x_{m+2}^1=a$. Let $\chi=\cap_{i=1}^{m+2}\mathcal{D}_{m}(x_i)$. If $|\chi^a|=2^m$, then 
$$|\chi^{\bar{a}}| \leq 2^{m+1}-m-2.$$

We need the following claim. 
\begin{claim} \label{claim:contians very binary seq}
Suppose a fixed $u \in \{0,1\}^{2m}$ contains every binary sequences of length $m$ as a subsequence. Then $$u=b_1\bar{b}_1b_2\bar{b}_2\cdots b_m\bar{b}_m$$
for some  $b_1,b_2,\cdots b_m \in \{0,1\}$. 
\end{claim}
\begin{proof}
Since both $[0]^m$ and $[1]^m$ are the subsequences of $u$, $u$ has to contain exactly $m$ zeros and $m$ ones.  

The first two symbols of $u$ can not be the same. Suppose they are both $b$. Then $[b][\bar{b}]^{m-1}$ is not a subsequence of $u$. Indeed,assume that the first two symbols of $u$ are $b$. Then there are only $m-2$ positions left that can be $b$ if the third symbol is $\bar{b}$. Let $u:=u^1u^2\cdots u^{2m}$. Thus $u^2=\overline{u^1}$.

Now claim that $u^3u^4\cdots u^{2m}$ contains every binary sequence $z$ of length $m-1$. Indeed, $u^2z$ is a binary sequence of length $m$ and so it is a subsequence of $u$. Since $u^2$ occurs no earlier than position 2, $z$ is a subsequence of $u^3u^4\cdots u^{2m}$. Now we can induct on $m$ and this completes the proof of the claim.
\end{proof}

We can represent $x_i$ for $i \in [m+2]$ as follows.
$$x_i=au_i$$ where $u_1,u_2,\cdots, u_{m+2}$ are $m+2$ binary sequences of length $2m$ such that for every $i \in [m+2]$
$$\mathcal{D}_{m}(u_i)=\{0,1\}^m.$$
Therefore, by Claim ~\ref{claim:contians very binary seq}, 
$$u_i=b_i^1\bar{b}_i^1b_i^2\bar{b}_i^2..\cdots b_i^m \bar{b}_i^m.$$
Let us write $b_i=b_i^1b_i^2\cdots b_i^m.$ Since $u_i$'s are distinct, $b_i$'s are distinct.

For each $i \in [m+2]$, consider $y_i:=\bar{a}b_i$. We  have the following claim.
\begin{claim}
For $i \in [m+2]$,
$$y_i \notin \mathcal{D}_{m-1}(u_i).$$ 
\end{claim} 
\begin{proof}
We will prove it by contradiction. Suppose the conclusion is not true. Then let $p_0<p_1<\cdots <p_m$ be such that 
$$u_i^{p_0}=\bar{a}$$
and for $j \in [m]$
$$u_i^{p_j}=b_i^j.$$

We claim that $p_j \geq 2j+1$ 
for $j \in [m]$. First consider the case where $j=1$. If $u_i^1=\bar{a}$, then $b_i^1$ can not occur before the third coordinate of $u_i$ because $u_i^2=\bar{b}_i^1$. If $u_i^j=\bar{a}$ for $j >1$, then it is obvious that $b_i^1$ can not occur before the third coordinate of $u_i$. Thus we have 
$$p_1 \geq 3.$$
Suppose that $p_{j-1} \geq 2j-1.$ Note that the element in the $2j$'s coordinate of $u_i$ is $\bar{b}_i^j$. Thus $b_i^{j}$ can not occur before the $2j+1$'s coordinate of $u_i$. Therefore, 
$$p_j \geq 2j+1.$$

Now take $j=m$ and we have $p_m \geq 2m+1$. However, the length of $u_i$ is $2m$, which is a contradiction. This completes the proof of the claim.
\end{proof}

Note that $y_1,y_2 \cdots, y_{m+2}$ are all distinct since $b_i$'s are distinct. Since all $y_i$ for $i \in [m+2]$ start with $\bar{a}$, none of them is in $\chi^{\bar{a}}$. Since there $2^m$ binary sequences of length $m+1$ starting with $\bar{a}$, we have 
$$|\chi^{\bar{a}}| \leq 2^m-m-2. \qedhere $$ 
\end{proof}

The next lemma corresponds to \textbf{Subcase 1}. 
\begin{lemma}
Suppose that among the vectors $x_1, x_2, \dots, x_{\ell}$, 
at least one of them begins with $11$ and at least one of them begins with $00$. Then
$$|\mathcal{D}_t(x_1) \cap \mathcal{D}_t(x_2) \cap ... \cap \mathcal{D}_t(x_{\ell}) |\leq N_{\ell}(n,t).$$
\end{lemma}
\begin{proof}
Without loss of generality, we can assume that $x_1$ starts with $11$ and $x_2$ starts with $00$. Let $\chi=\cap_{i=1}^{\ell}\mathcal{D}_t(x_i)$. By Lemma ~\ref{lemma 1}, suppose that $x_1^k$ with $k 
\geq 3$ is the first occurrence of the symbol $0$ and then 
we have 
$$|\chi^0| \leq |0 \circ D_{t-(k-1)}(x_1^{k+1},x_1^{k+2},\cdots,x_1^{n})| \leq D(n-k,t-k+1) \leq D(n-3,t-2).$$
Similarly, we have $$|\chi^{1}| \leq D(n-3,t-2). $$
Thus $|\chi| \leq 2D(n-3,t-2).$
Note that 
$$N_\ell(n,t)=D(n-2,t-1)+N_{\ell-1}(n-2,t-1).$$
By Lemma ~\ref{lemma: D(x-1,y-1)}, we have 
$$N_{\ell-1}(n-2,t-1) \geq D(n-3,t-2).$$
Note that $D(n-3,t-2)\leq D(n-2,t-1)$. Thus we have 
$$2D(n-3,t-2) \leq N_\ell(n,t). \qedhere$$
\end{proof}

The next lemma corresponds to \textbf{Subcase 2}.

\begin{lemma}
Suppose that among the vectors $x_1, x_2, \dots, x_{\ell}$, 
at least one of them begins with $11$ and at least one of them begins with $01$. Then
$$|\mathcal{D}_t(x_1) \cap \mathcal{D}_t(x_2) \cap ... \cap \mathcal{D}_t(x_{\ell}) |\leq N_{\ell}(n,t).$$
\end{lemma}
\begin{proof}
Without loss of generality, we can assume that $x_1$ starts with $11$ and $x_2$ starts with $01$. Let $\chi= \cap_{i=1}^{\ell} \mathcal{D}_{t}(x_i)$. 

By Lemma ~\ref{lemma 1}, suppose that $x_1^k$ with $k 
\geq 3$ is the first occurrence of the symbol $0$ and then 
we have 
$$|\chi^0| \leq |0 \circ D_{t-(k-1)}(x_1^{k+1},x_1^{k+2},\cdots,x_1^{n})| \leq D(n-k,t-k+1) \leq D(n-3,t-2).$$

By Lemma ~\ref{lemma 1}, we have 
$$|\chi^1| \leq |1 \circ D_{t-1}(x_2^3x_2^4\cdots x_2^{n})| \leq D(n-2,t-1).$$
Thus by Lemma ~\ref{lemma: D(x-1,y-1)}, we have 
\begin{align*}
|\chi| &=|\chi^0|+|\chi^1|\\
& \leq D(n-3,t-2)+D(n-2,t-1)\\
& \leq N_{\ell-1}(n-2,t-1)+D(n-2,t-1)\\
&= N_{\ell}(n,t). \qedhere
\end{align*}
\end{proof}

The next lemma corresponds to \textbf{Subcase 3}.
\begin{lemma} \label{lemma: l_1 10,l_2 01}
Suppose that among the vectors $x_1, x_2, \dots, x_{\ell}$, 
$\ell_1$ of them begin with $10$ and $\ell_2$ of them begin with $01$, 
where $\ell_1 \geq 1$, $\ell_2 \geq 1$, and $\ell_1+\ell_2=\ell$. Then
$$|\mathcal{D}_t(x_1) \cap \mathcal{D}_t(x_2) \cap ... \cap \mathcal{D}_t(x_{\ell}) |\leq N_{\ell}(n,t).$$
\end{lemma}
\begin{proof}
First let us assume that $\ell_1 \geq 2$ and $\ell_2 \geq 2$. Let $\chi= \cap_{i=1}^{\ell} \mathcal{D}_{t}(x_i)$ and then 
$\chi^0=0 \circ \cap_{i=1}^{\ell_1}\mathcal{D}_{t-1}(x_i^3x_i^4\cdots x_i^n)$.
\begin{equation}
 |\chi^0| \leq \max_{y_i \in \{0,1\}^{n-2}} \cap_{i=1}^{\ell_1} \mathcal{D}_{t-1}(y_i) \leq N_{\ell_1}(n-2,t-1). \label{eq:N_l_1(n-2,t-1)}   
\end{equation}
Similarly, we have 
\begin{equation}
|\chi^1| \leq \max_{y_i \in \{0,1\}^{n-2}} \cap_{i=1}^{\ell_2} \mathcal{D}_{t-1}(y_i) \leq N_{\ell_2}(n-2,t-1). \label{eq:N_l_2(n-2,t-1)}
\end{equation}

Note that since  $n\geq t+\ell-1$, we have $n-2 \geq t-1+\ell_i-1$ for $i \in \{1,2\}$. Therefore, \eqref{eq:N_l_1(n-2,t-1)} and \eqref{eq:N_l_2(n-2,t-1)} are valid from the induction hypothesis. 

We will show that 
$$
N_{\ell}(n,t) \ge N_{\ell_1}(n-2,t-1)+N_{\ell_2}(n-2,t-1).
$$

We first note the identity
\begin{equation}\label{eq:D-recursion}
D(n,t)=D(n-1,t) + D(n-2,t-1).
\end{equation}
Since
$$
D(n-1,t)-D(n-2,t-1)
=\binom{n-t-1}{t}\ge 0,
$$
we obtain from \eqref{eq:D-recursion} that
\begin{equation}\label{eq:D-ineq}
D(n,t)\ge 2D(n-2,t-1).
\end{equation}

Applying \eqref{eq:D-ineq} to each term in the definition of $N_\ell(n,t)$ gives
\begin{align}
N_{\ell}(n,t)
&=\sum_{i=1}^{\ell-2} D(n-2i,t-i)
\;+\;2D(n-2(\ell-1),t-(\ell-1))\\
&\ge 2\sum_{i=1}^{\ell-2} D(n-2-2i,t-1-i)
\;+\;2D(n-2-2(\ell-1),t-\ell).\label{eq:N-ineq}
\end{align}

Now we expand $N_{\ell_1}(n-2,t-1)$ and $N_{\ell_2}(n-2,t-1)$:
\begin{align*}
N_{\ell_1}(n-2,t-1)
&=\sum_{i=1}^{\ell_1-2} D(n-2-2i,t-1-i)
\;+\;2D(n-2-2(\ell_1-1),t-\ell_1),\\
N_{\ell_2}(n-2,t-1)
&=\sum_{j=1}^{\ell_2-2} D(n-2-2j,t-1-j)
\;+\;2D(n-2-2(\ell_2-1),t-\ell_2).
\end{align*}

Summing, we obtain
\begin{align}
&N_{\ell_1}(n-2,t-1)+N_{\ell_2}(n-2,t-1)\\
&=\sum_{i=1}^{\ell_1-2} D(n-2-2i,t-1-i)
+\sum_{j=1}^{\ell_2-2} D(n-2-2j,t-1-j)\\
&+2D(n-2-2(\ell_1-1),t-\ell_1)
+2D(n-2-2(\ell_2-1),t-\ell_2).\label{eq:sum-l1l2}
\end{align}

 Also we have 
$$D(n-2(\ell_1-1),\,t-(\ell_1-1)) \geq 2D(n-2-2(\ell_1-1),\,t-\ell_1), $$

and 
$$D(n-2(\ell_2-1),\,t-(\ell_2-1)) \geq  2D(n-2-2(\ell_2-1),\,t-\ell_2).$$

Since $\ell_1-1,\ell_2-1\le \ell-2$ when $\ell_1,\ell_2\ge 2$, every term in \eqref{eq:sum-l1l2} appears among the summands in \eqref{eq:N-ineq} with coefficient at most $2$. The additional term $2D(n-2-2(\ell-1),t-\ell)\ge 0 $ in \eqref{eq:N-ineq} only increases the right-hand side of that inequality.

Therefore, comparing \eqref{eq:N-ineq} and \eqref{eq:sum-l1l2} yields
$$
N_{\ell}(n,t)\ge N_{\ell_1}(n-2,t-1)+N_{\ell_2}(n-2,t-1).$$

Now assume that $\ell_1=1$ and $\ell_2=\ell-1$. Then 
$$|\chi^{0}| \leq D(n-2,t-1)$$ and $$|\chi^{1}| \leq N_{\ell-1}(n-2,t-1)$$ by induction hypothesis. 

Thus we have 
$$|\chi| \leq D(n-2,t-1)+N_{\ell-1}(n-2,t-1)=N_{\ell}(n,t). \qedhere$$

\end{proof}

Now the induction step is finished. we have the following theorem. 
\begin{theorem}
For $n \geq t+\ell-1$, $t\geq 1$, and $\ell \geq 3$,  we have
$$N(n,\ell,t) \leq N_{\ell}(n,t).$$
\end{theorem}

\section{Conclusion and Open Problems}
In this paper we derive the maximum size of the intersection of the $t$-deletion balls centered at $x_1,x_2,\cdots,x_{\ell}$ where $x_i \in \{0,1\}^n$ for $i \in \{1,2,\cdots,j\}$ and $x_i \neq x_j$
for $i \neq j$. Therefore, the minimum number of the deletion channel outputs is $N_{\ell}(n,t)+1$ in order to reconstruct a list of size $\ell-1$ of candidate sequences, one of which corresponds to the original transmitted sequence. 

Since our results concern the binary case, the first open question is the following.  

\medskip
\noindent
\textbf{Question 1.} 
What is the maximum possible size of the intersection of the $t$-deletion balls 
centered at $x_1, x_2, \dots, x_{\ell}$, where 
$x_i \in \{0,1,\dots,q-1\}^n$ for all $i \in \{1,2,\dots,\ell\}$ and 
$x_i \neq x_j$ for $i \neq j$?

\medskip
Since the only condition above is $x_i \neq x_j$ for $i \neq j$, the following naturally arises as a second open question. 

\medskip
\noindent
\textbf{Question 2.} 
What is the maximum possible size of the intersection of the $t$-deletion balls 
centered at $x_1, x_2, \dots, x_{\ell}$, where 
$x_i \in \{0,1,\dots,q-1\}^n$ for all $i \in \{1,2,\dots,\ell\}$ and 
$d_L(x_i, x_j) \geq d$ for $i \neq j$?

\section*{Acknowledgment}
The author would like to thank Prof.~Alexander Barg for suggesting this problem and providing valuable guidance. Additionally, the author extends thanks to  Yihan Zhang for insightful discussions. 
\section*{AI Disclosure}
ChatGPT 5.6 Plus was used in the preparation of this paper. In particular, during discussions about the proof of Lemma ~\ref{lemma:x_1=x_2=a}, especially the boundary case $n=t+\ell-1$, ChatGPT suggested several approaches including the strategy adopted in the paper of dividing this part of the proof into three cases. All mathematical arguments and claims in this paper were developed by the author, who bears full responsibility for the correctness of the results.

\bibliographystyle{plain}
\bibliography{main.bib}

\end{document}